\documentclass[letterpaper,journal,twocolumn]{IEEEtran}
\usepackage{amsmath,amssymb,amsthm}
\usepackage{flushend}
\usepackage{cite}
\usepackage{graphicx}

\usepackage[usenames]{color}
\definecolor{plum}  {rgb}{.4,0,.4}
\definecolor{BrickRed} {rgb}{0.6,0,0}
\usepackage[plainpages=false,pdfpagelabels,colorlinks=true,linkcolor=BrickRed,citecolor=plum]{hyperref}

\def\deq{\triangleq}
\def\R{\mathbb{R}}
\def\N{\mathbb{N}}
\def\eps{\varepsilon}

\def\E{\mathbb{E}}
\def\Pr{{\mathbb P}}
\def\emp{{\mathsf P}}
\def\tvar#1{\| #1 \|_{{\rm TV}}}
\def\dvar#1#2{\tvar{#1-#2}}

\def\cA{{\cal A}}
\def\cB{{\cal B}}

\def\cE{{\cal E}}
\def\cF{{\cal F}}
\def\cG{{\cal G}}
\def\cP{{\cal P}}
\def\cS{{\cal S}}
\def\cT{{\cal T}}
\def\cU{{\mathsf U}}
\def\cV{{\mathsf V}}
\def\cW{{\mathsf W}}
\def\cX{{\mathsf X}}
\def\cY{{\mathsf Y}}
\def\cZ{{\mathsf Z}}

\def\wh#1{\hat{#1}}
\def\td#1{\tilde{#1}}

\newtheorem{example}{Example}
\newtheorem{theorem}{Theorem}
\newtheorem{definition}{Definition}
\newtheorem{proposition}{Proposition}
\newtheorem{remark}{Remark}
\newtheorem{lemma}{Lemma}
\newtheorem{fact}{Fact}

\begin{document}

\title{Empirical Processes, Typical Sequences\\
and Coordinated Actions in Standard Borel Spaces}

\author{Maxim Raginsky,~\IEEEmembership{Member,~IEEE}
\thanks{A preliminary version of this work was presented at the IEEE International Symposium on Information Theory, Austin, TX, July 2010.}
\thanks{M.~Raginsky was with the Department of Electrical and Computer Engineering, Duke University, Durham, NC 27708, USA. He is now with the Department of Electrical and Computer Engineering and the Coordinated Science Laboratory, University of Illinois, Urbana, IL 61801, USA. E-mail: maxim@illinois.edu.}}

\maketitle
\thispagestyle{empty}


\begin{abstract}
This paper proposes a new notion of typical sequences on a wide class of abstract alphabets (so-called standard Borel spaces), which is based on approximations of memoryless sources by empirical distributions uniformly over a class of measurable ``test functions." In the finite-alphabet case, we can take all uniformly bounded functions and recover the usual notion of strong typicality (or typicality under the total variation distance). For a general alphabet, however, this function class turns out to be too large, and must be restricted. With this in mind, we define typicality with respect to any Glivenko--Cantelli function class (i.e., a function class that admits a Uniform Law of Large Numbers) and demonstrate its power by giving simple derivations of the fundamental limits on the achievable rates in several source coding scenarios, in which the relevant operational criteria pertain to reproducing empirical averages of a general-alphabet stationary memoryless source with respect to a suitable function class.\\ \\
\begin{IEEEkeywords}Coordination via communication, empirical processes, Glivenko--Cantelli classes, rate distortion, source coding, standard Borel spaces, typical sequences, uniform laws of large numbers.
\end{IEEEkeywords}
\end{abstract}

\section{Introduction}
\label{sec:intro}

\PARstart{T}{he} notion of {\em typical sequence} has been central to information theory since Shannon's original paper \cite{Shannon48}. For finite alphabets, it leads to simple and intuitive proofs of achievability in a wide variety of source and channel coding settings, including multiterminal scenarios \cite{CovTho06}. Another appealing aspect of typical sequences is that they provide a language for {\em approximation} of information sources in total variation distance using finite communication resources. Recent work of Cuff et al.~\cite{CufPerCov09} on coordination via communication serves as a particularly striking example of the power of this language. 

For abstract alphabets, however, most of this power is lost; while such  results as the asymptotic equipartition property carry over \cite{Bar85}, in most other situations, particularly involving lossy codes, one has to resort to ergodic theory \cite{Gra90a} or large deviations theory \cite{DemZei98}. Direct approximation of abstract memoryless sources in total variation using empirical distributions is, in general, impossible (cf.~Sec.~\ref{sec:general_typicality} for details). However, it is precisely this direct approximation that renders typicality-based proofs of achievability so transparent.

The present paper makes two contributions. First, we propose a way to revise the notion of typicality for {\em general} alphabets (more specifically, standard Borel spaces \cite{Preston,Gra09}), allowing for similarly transparent achievability arguments. When two probability measures are close in total variation, the corresponding expectations of {\em any} bounded measurable function are also close. For general alphabets, when one of the measures is discrete, this is too much to ask. Instead, we advocate an approach based on suitably {\em restricting} the class of functions on which we would like to match statistical expectations with sample (empirical) averages. Provided the Law of Large Numbers holds {\em uniformly} over the restricted function class, we can speak of typical sequences {\em with respect to this class} and develop typicality-based achievability arguments in close parallel to the finite-alphabet case. The central object of study is the {\em empirical process} \cite{Pol84,WaaWel96,vanDeGeer00} indexed by the function class, which gives information on the deviation of empirical means from statistical means for a given realization of the source under consideration, and the total variation distance is replaced by the supremum norm of this empirical process.

The second contribution consists of applying our new notion of typicality to several source coding problems which, following the terminology of \cite{CufPerCov09}, can be thought of as ``empirical coordination'' of actions in a two-node network. Roughly speaking, the objective is to use communication resources in order to reproduce (or approximate) the empirical distribution of a given source sequence, rather than the sequence itself, with or without side information. This coordination viewpoint suggests a new operational framework suitable for problems involving distributed learning, control, and sensing.

\subsection{Preview of the results}

\begin{figure}
	\centerline{\includegraphics[width=0.8\columnwidth]{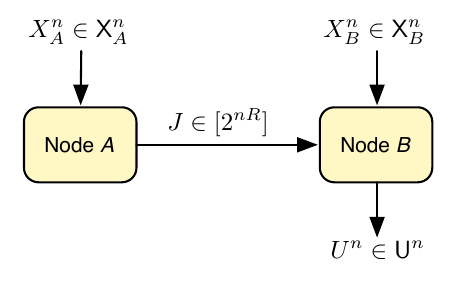}}
	\caption{\label{fig:generic} Empirical coordination of actions in a two-node network. Node $A$ (resp., $B$) observes a random $n$-tuple $X^n_A$ (resp., $X^n_B$), where $(X_{A,1},X_{B,1}),\ldots,(X_{A,n},X_{B,n})$ are i.i.d.\ pairs of correlated random variables. A message is sent from Node $A$ to Node $B$ at rate $R$ to specify the $n$-tuple $U^n$.}
\end{figure}

Consider the two-node network shown in Figure~\ref{fig:generic}. There is an alphabet $\cX_A$ associated with Node $A$, and two alphabets, $\cX_B$ and $\cU$, associated with Node $B$. Initially, Node $A$ (resp., Node $B$) observes a random $n$-tuple $X^n_A \in \cX^n_A$ (resp., $X^n_B \in \cX^n_B$), where the pairs $(X_{A,1},X_{B,1}),\ldots,(X_{A,n},X_{B,n})$ are i.i.d.\ draws from some specified probability law $P_{X_A X_B}$ on $\cX_A \times \cX_B$. We also have a target conditional probability law $P_{U|X_A}$ on $\cU$ given $X_A$. Node $A$, given its knowledge of $X^n_A$, $P_{X_A X_B}$, and $P_{U|X_A}$, communicates some information $J$ to Node $B$ at rate $R$. The latter receives $J$ and, using its knowledge of $X^n_B$, $P_{X_A X_B}$, and $P_{U|X_A}$, generates an $n$-tuple $U^n \in \cU^n$.

Now imagine that there is an external observer with access to $X^n_A$ and $U^n$, who also knows $P_{X_A}$ and $P_{U|X_A}$. This observer has a collection $\cF$ of ``test functions'' $f : \cX_A \times \cU \to [-1,1]$ and can compute the {\em empirical expectation} (or sample average) $n^{-1}\sum^n_{i=1}f(X_{A,i},U_i)$ and the ``true'' expectation $\E f(X_A,U)$ w.r.t.\ the joint law $P_{X_A U} = P_{X_A} \otimes P_{U|X_A}$ for any $f \in \cF$. We assume that Nodes $A$ and $B$ know $\cF$, but do not know which $f \in \cF$ the observer will pick. The objective is then to minimize the expected worst-case deviation between the empirical expectations and the true expectations:
\begin{align*}
	\text{minimize }\,\,\, \E \sup_{f \in \cF} \left| \frac{1}{n}\sum^n_{i=1}f(X_{A,i},U_i) - \E f(X_A,U)\right|
\end{align*}
over all admissible encoding and decoding strategies given the rate constraint $R$ and the information patterns at the two nodes (i.e., which node knows what). In other words, the goal is to ensure that, from the observer's viewpoint, the empirical distribution of $\{(X_{A,i},U_i)\}^n_{i=1}$ is as close as possible to the target distribution $P_{X_A  U}$ in the sense that the corresponding expectations of all $f \in \cF$ are as close as possible, {\em uniformly} over $\cF$. Operational criteria of this kind arise, e.g., in the context of statistical learning from random samples \cite{BueKum96,RaginskyITA09}, where the functions in $\cF$ may be viewed as the losses of various candidate predictors of $U$ given $X_A$.

In this paper, we consider two special cases of this set-up:
\begin{enumerate}
	\item Given two alphabets $\cX$ and $\cY$, we take $\cX_A = \cX$, $\cX_B = \varnothing$, $\cU = \cY$. This is a generalization of the basic two-node empirical coordination problem \cite[Section~III.C]{CufPerCov09} to abstract alphabets. The work of \cite{CufPerCov09} is, in turn, related to the problem of communication of probability distributions \cite{KraSav07}. (A related problem, though with a slightly different operational criterion, is lossy source coding with respect to a family of distortion measures \cite{DemWei03}.) 
	\item We have $\cX_A$ and $\cU$ as above, and also $\cX_B = \cZ$, where $\cZ$ is some third alphabet. This is a generalization of the problem stated in 1), but now we also allow side information at the decoder.
\end{enumerate}

Our achievability results hinge on the assumption that the function class $\cF$ admits the {\em Uniform Law of Large Numbers (ULLN)}. Given an abstract alphabet $\cZ$, we say that a class $\cF$ of functions $f : \cZ \to [-1,1]$ admits the ULLN if the following holds: for any i.i.d.\ random process $\{Z_i\}^\infty_{i=1}$ over $\cZ$, we have
\begin{align*}
	\sup_{f \in \cF} \left| \frac{1}{n} \sum^n_{i=1} f(Z_i) - \E f(Z_1) \right| \xrightarrow{n \to \infty} 0, \quad \text{a.s.}
\end{align*}
The quantity inside the $|\cdot |$ is referred to as the {\em empirical process} associated with $Z^n$, and describes the fluctuations of the sample mean of each $f$ around its expectation. We define an $n$-tuple $z^n = (z_1,\ldots,z_n) \in \cZ^n$ to be {\em $\eps$-typical} w.r.t.\ $\cF$ for a probability law $P$ if
\begin{align*}
	\sup_{f \in \cF} \left| \frac{1}{n}\sum^n_{i=1}f(z_i) - \E_P f(Z) \right| < \eps.
\end{align*}
Turning now to the set-up of Figure~\ref{fig:generic}, let us assume that the observer's function class $\cF$ satisfies the ULLN. Then a simple achievability argument exploits the fact (which we prove under mild regularity conditions) that, for any probability law $Q = Q_{X_A X_B U}$ under which $X_B \to X_A \to U$ is a Markov chain, there exist a rate-$R$ encoding $\wh{U}^n(X^n_A)$ from $\cX^n_A$ into $\cU^n$ and a deterministic mapping $g$ from $\cX^n_B$ into $\cX^n_A$, such that the tuple $(g(X_B^n),\wh{U}^n)$ is $\eps$-typical w.r.t.\ $\cF$ for $Q$, provided $R > I(X_A; U | X_B)$. When $\cX_B = \varnothing$, we simply apply the above argument to ``degenerate'' Markov chains of the form $X_A \to X_A \to U$, where the rate condition becomes $R > I(X_A;U)$.

We list the salient features of our approach:
\begin{itemize}
	\item When the underlying alphabet $\cZ$ is finite, the ULLN is satisfied by the class of {\em all} functions $f : \cZ \to [-1,1]$, and our definition of typicality reduces to strong typicality \cite{CufPerCov09,CovTho06}.
	\item When $\cZ$ is a complete separable metric space, the ULLN is satisfied by the class of all Lipschitz functions $f : \cZ \to [-1,1]$ with $\| f \|_\infty \le 1$ and Lipschitz constant bounded by $1$. Moreover, the ULLN in this case is equivalent to almost sure weak convergence of empirical distributions (Varadarajan's theorem \cite[Theorem~11.4.1]{Dud02}).
	\item In general, there is a veritable plethora of function classes satisfying the ULLN (we present several examples in Section~\ref{ssec:GC_examples}). For instance, when $\cZ = \R^d$, the ULLN is satisfied by the indicator functions of all halfspaces, balls, or rectangles (and of finite unions or intersections thereof). One example, particularly relevant in source coding, is the collection of indicator functions of Voronoi cells induced by an arbitrary set of $m$ points in $\R^d$, for any fixed $m$ --- indeed, any such cell is an intersection of $O(m)$ halfspaces. Hence, our results apply to the setting where $\cX_A \times \cU \subseteq \R^d$ and each $(X_{A,i},U_i)$ is observed through an $m$-point nearest-neighbor quantizer.
\end{itemize}

\subsection{Related work}

The focus of the present paper is exclusively on source coding. However, a recent preprint of Mitran \cite{Mitran} uses weak convergence to develop an extension of typical sequences to Polish alphabets and then applies that definition to several channel coding problems, including an achievability result for Gel'fand--Pinsker channels \cite{GelPin80} with input cost constraints. What distinguishes Mitran's work from ours is his careful use of several equivalent characterizations of weak convergence via the portmanteau theorem \cite[Theorem~11.1.1]{Dud02}. In particular, his approach requires an explicit construction of a countable generating set for the underlying Borel $\sigma$-algebra that consists of the continuity sets of the probability law of interest. As a consequence, he is able to establish a generalization of the Markov lemma \cite{Ber78,Kra07}, which in turn allows him to use binning just like in the finite-alphabet case. By contrast, our notion of typicality is considerably broader (and, in fact, contains the one based on weak convergence as a special case), but, since we do not make any major structural assumptions beyond those needed for the ULLN, we cannot establish anything as strong as the Markov lemma. However, our proof technique does not rely on the Markov lemma in its strong form, and is more in the spirit of Wyner and Ziv \cite{Wyn75,WynZiv76,Wyn78}.

We also note that a restricted notion of typicality based on weak convergence was used by Kontoyiannis and Zamir \cite{KonZam06} in the context of universal vector quantization using entropy codes. The idea there is to consider sequences of increasing length, whose empirical distributions converge in the weak topology to the output distribution of an optimal test channel in a Shannon rate-distortion problem.

\subsection{Contents of the paper}

The remainder of the paper is organized as follows. Section~\ref{sec:prelims} sets up the notation and lists the preliminaries. In Section~\ref{sec:GC} we formally define function classes that satisfy the ULLN and give several examples. Then, in Section~\ref{sec:general_typicality} we motivate and formally describe our approach to typicality and establish a number of key properties, including a lemma on the preservation of typicality in a Markov structure. Next, in Section~\ref{sec:example_apps}, using this lemma as the main technical tool, we illustrate the power of the proposed new approach by proving three theorems concerning fundamental limits on minimal achievable rates for (i) two-node empirical coordination; (ii) two-node empirical coordination with side information at the decoder; and (iii) lossy source coding under a family of distortion measures. Although these results apply to general (uncountably infinite) alphabets, the proofs are as intuitive and simple as in the finite-alphabet scenario. We follow up with some concluding remarks in Section~\ref{sec:conclusion}. Lengthy proofs and discussions of auxiliary technical results are relegated to the Appendices.

\section{Preliminaries and notation}
\label{sec:prelims}

All spaces in this paper are assumed to be {\em standard Borel spaces} (for detailed treatments, see the lecture notes of Preston \cite{Preston} or Chapter~4 of Gray \cite{Gra09}):

\begin{definition} A measurable space $(\cZ,\cB_\cZ)$ is {\em standard Borel} if it can be metrized with a metric $d$ such that (1) $(\cZ,d)$ is a complete separable metric space, and (2) $\cB_\cZ$ coincides with the Borel $\sigma$-algebra of $(\cZ,d)$ (the smallest $\sigma$-algebra containing all open sets).
\end{definition}

\begin{remark} {\em A Polish space (i.e.,~a separable topological space whose topology can be metrized with a complete metric) is automatically standard Borel. In fact, the most general known class of standard Borel spaces consists of Borel subsets of Polish spaces \cite[Theorem~4.3]{Gra09}.}
\end{remark}

From now on, when dealing with a (standard Borel) space $\cZ$, we will often not mention its Borel $\sigma$-algebra explicitly. In particular, we will tacitly assume that all probability measures on $\cZ$ are defined w.r.t.\ $\cB_\cZ$. The main objects associated with $\cZ$ that are of interest to us are as follows:
\begin{itemize}
	\item $\cP(\cZ)$ is the space of all probability measures on $\cZ$
	\item $M(\cZ)$ is the space of all measurable functions $f : \cZ \to \R$
	\item $M^b(\cZ) \subset M(\cZ)$ is the normed space of all bounded measurable functions $f : \cZ \to \R$ with the sup norm
	\begin{align*}
		\| f \|_\infty \deq \sup_{z \in \cZ} |f(z)|
	\end{align*}
	\item $M^{b,1}(\cZ) \deq \left\{ f \in M^b(\cZ) : \| f \|_\infty \le 1 \right\}$.
\end{itemize}
Other notation will be introduced as needed.

Standard Borel spaces possess just enough useful structure for our purposes. In particular, their $\sigma$-algebras are countably generated and contain all singletons. They also admit the existence of regular conditional distributions: If $\cZ = \cX \times \cY$ with the product $\sigma$-algebra, then the probability law $P \in \cP(\cZ)$ of any random couple $(X,Y) \in \cZ$ can be disintegrated as
\begin{align*}
	P(A \times B) = \int_A P_{Y|X}(B|x) P_X(dx),  \forall A \in \cB_\cX, B \in \cB_\cY
\end{align*}
where $P_X \in \cP(\cX)$ is the marginal distribution of $X$ and $P_{Y|X}(\cdot|\cdot) : \cB_\cY \times \cX \to [0,1]$ is a {\em Markov kernel}, i.e., $P_{Y|X}(\cdot|x) \in \cP(\cY)$ for all $x \in \cX$ and $P_{Y|X}(B|\cdot) \in M(\cX)$ for all $B \in \cB_\cY$. Given a random triple $(U,X,Y) \in \cU \times \cX \times \cY$ with joint law $P \in \cP(\cU \times \cX \times \cY)$, we will say that they form a {\em Markov chain} in that order (and write $U \to X \to Y$) if
\begin{align*}
	P_{U|XY}(A|x,y) = P_{U|X}(A|x), \qquad \forall A \in \cB_\cU
\end{align*}
for $P$-almost all $x,y$.

We will often use de Finetti's linear functional notation for expectations \cite[Section~1.4]{PollardUGMP}. That is, for any $P \in \cP(\cZ)$ and a $P$-integrable function $f: \cZ \to \R$,
\begin{align*}
P(f) \deq \E_P f(Z) \equiv \int_\cZ f dP,
\end{align*}
and we will extend this notation in an obvious way to integrals with respect to signed Borel measures on $\cZ$. Given a class $\cF$ of measurable functions $f \in M^{b,1}(\cZ)$, we can define a seminorm on the space of all signed measures on $\cZ$ via
\begin{align*}
\| \nu  \|_\cF \deq \sup_{f \in \cF} |\nu(f) |.
\end{align*}
As an example, $\| P - P' \|_{M^{b,1}(\cZ)}$ is precisely the {\em total variation distance}
\begin{align}\label{eq:dvar}
\dvar{P}{P'} \deq 2 \sup_{A \in \cB_\cZ} |P(A) - P'(A)|
\end{align}
between $P,P' \in \cP(\cZ)$.

We will make use of several standard information-theoretic definitions \cite{Gra90a}. The {\em divergence} between $P$ and $P'$ in $\cP(\cZ)$ is defined as
\begin{align*}
D(P \| P') \deq \begin{cases}
P \left( \log (dP/dP') \right), & \text{if } P \ll P' \\
+ \infty, & \text{otherwise}
\end{cases}
\end{align*}
Given a $Q \in \cP(\cX \times \cY)$, the {\em mutual information} between $X \in \cX$ and $Y \in \cY$ with joint law $Q$ is
\begin{align*}
	I(Q) \deq D(Q \| Q_X \otimes Q_Y),
\end{align*}
	where $Q_X \otimes Q_Y$ is the product of the marginals. Whenever $Q$ is clear from context, we will also write $I(X;Y)$ instead of $I(Q)$. We will use standard notation for such things as the conditional mutual information.

\section{Uniform Laws of Large Numbers and Glivenko--Cantelli classes}
\label{sec:GC}

Given an $n$-tuple $z^n = (z_1,\ldots,z_n) \in \cZ^n$, let us denote by $\emp_{z^n}$ the induced {\em empirical measure}:
\begin{align*}
\emp_{z^n} \deq \frac{1}{n}\sum^n_{i=1} \delta_{z_i},
\end{align*}
where $\delta_{z_i}$ is the Dirac measure concentrated at $z_i$ (since $\cB_\cZ$ contains all singletons, $\delta_z \in \cP(\cZ)$ for every $z \in \cZ$). If $\{Z_i\}^\infty_{i=1}$ is an i.i.d.\ sequence with common distribution $P \in \cP(\cZ)$,  then the Strong Law of Large Numbers says that, for any $f \in M^{b,1}(\cZ)$, the empirical means
\begin{align*}
	\emp_{Z^n}(f) = \frac{1}{n} \sum^n_{i=1} f(Z_i), \qquad n \in \N
\end{align*}
converge to the true mean $P(f)$ almost surely. By the union bound, this holds for any {\em finite} family of functions. In this paper, we consider {\em infinite} function classes that admit a Uniform Law of Large Numbers --- that is, absolute deviations between empirical and true means converge to zero {\em uniformly} over the function class. The canonical example of such a class appears in the celebrated Glivenko--Cantelli theorem \cite[Theorem~11.4.2]{Dud02}: Let $Z$ be a real-valued random variable with CDF $F_Z$, and let $\{Z_i\}^\infty_{i=1}$ be an infinite sequence of i.i.d.\ copies of $Z$. For each $n$, consider the {\em empirical CDF}
\begin{align*}
	{\mathsf F}_{Z^n}(z) \deq \frac{1}{n}\sum^n_{i=1}1_{\{ Z_i \le z\}}.
\end{align*}
The Glivenko--Cantelli theorem then says that
\begin{align*}
	\sup_{z \in \R} |{\mathsf F}_{Z^n}(z) - F_Z(z) | \xrightarrow{n \to \infty} 0 \qquad \text{a.s.}
\end{align*}
To cast it as a statement about a function class, consider
\begin{align*}
	\cF \deq \left\{ f_z = 1_{(-\infty,z]} : z \in \R \right\}.
\end{align*}
Then for any $z \in \R$,
\begin{align*}
	{\mathsf F}_{Z^n}(z) &= \emp_{Z^n}(f_z) \\
	F_Z(z) &= P_Z(f_z)
\end{align*}
and consequently
\begin{align*}
	\sup_{z \in \R} |{\mathsf F}_{Z^n}(z) - F_Z(z)| &= \| \emp_{Z^n} - P \|_\cF \xrightarrow{n \to \infty} 0 \quad \text{a.s.}
\end{align*}
This motivates the following definition \cite{Pol84,WaaWel96,vanDeGeer00}:

\begin{definition}
A class $\cF$ of measurable functions $f\in M^{b,1}(\cZ)$ is called {\em Glivenko--Cantelli}\footnote{Strictly speaking, the proper term is ``universal Glivenko--Cantelli,'' but we will follow standard usage and just say ``Glivenko--Cantelli.''} (or GC, for short) if
\begin{align}\label{eq:GC}
\| \emp_{Z^n} - P \|_\cF \xrightarrow{n \to \infty} 0 \qquad \text{a.s.}
\end{align}
for every  $P \in \cP(\cZ)$, where $\{Z_i\}^\infty_{i=1}$ is an i.i.d.\ random process with marginal distribution $P$. 
\end{definition}

\begin{remark}{\em In view of this definition, the classical Glivenko--Cantelli theorem can be paraphrased as follows: The class of all indicator functions of semi-infinite intervals of the form $(-\infty,z]$, $z \in \R$, is GC.}\end{remark}

\begin{remark}{\em The restriction to bounded functions is mostly needed for technical convenience and can be removed by means of suitable moment conditions and straightforward, though tedious, truncation arguments. A nice side benefit of the boundedness assumption, though, is that no loss of generality occurs if the almost sure convergence in \eqref{eq:GC} is replaced with convergence in probability \cite{Steele78,WaaWel96}.}\end{remark}

\begin{remark}\label{rem:measurability}{\em It should be borne in mind that when the function class $\cF$ is uncountable, $\| \emp_{Z^n} - P \|_\cF$ may not be a random variable (there is always a risk of spawning a nonmeasurable monster whenever one dabbles in uncountable operations). There are a number of ways to deal with such issues, as detailed in \cite[Appendix]{Pol84} or \cite[Section~2.3]{WaaWel96}. For our purposes, it will suffice to assume that $\cF$ is countable or ``nice'' in the sense that it contains a countable subset $\cG$ such that for every $f \in \cF$ there is a sequence $\{g_m\}$  in $\cG$ converging to $f$ pointwise. Then
	\begin{align*}
		\| \emp_{Z^n} - P \|_\cF = \| \emp_{Z^n} - P \|_\cG,
		\end{align*}
and the r.h.s.\ is a measurable function of $Z^n$ \cite[p.~110]{WaaWel96}.}
\end{remark}

Let $(\Omega,\cB,\Pr)$ be an underlying probability space for the random process $\{Z_i\}$. Then for each $n$ we can construct another random process on $(\Omega,\cB,\Pr)$, indexed by $\cF$:
\begin{align*}
\Delta^{(n)}_f(\omega) \deq \emp_{Z^n(\omega)}(f) - P(f), \qquad f \in \cF.
\end{align*}
This is an instance of an {\em empirical process} \cite{Pol84,WaaWel96,vanDeGeer00}, which is used to describe the fluctuations of the empirical means $\emp_{Z^n}(f)$ around the expectation $P(f)$. A GC class is one for which the $\ell^\infty(\cF)$ norms
\begin{align*}
	\big\| \Delta^{(n)}_f(\omega) \big\|_\cF = \sup_{f \in \cF} \big|\Delta^{(n)}_f(\omega)\big|
\end{align*}
of the empirical processes $\{\Delta^{(n)}_f \}_{f \in \cF}$, $n \ge 1$, converge to zero almost surely. 

\subsection{Examples of Glivenko--Cantelli classes}
\label{ssec:GC_examples}

We close this section by listing several examples of GC classes. Usually, whether or not a given class $\cF$ is GC depends on how ``large'' it is. The simplest notion of size is captured by the (metric) {\em entropy numbers} of $\cF$ \cite{KolTih61}. Given any $\eps >0$, the covering number $N(\eps,\cF,Q)$ of $\cF \subset M^{b,1}(\cZ)$ w.r.t.\ a probability measure $Q \in \cP(\cZ)$ is the minimal number of balls $\{ g : \| g - f \|_{L^1(Q)} \le \eps \}$, $f \in M^{b,1}(\cZ)$, of radius $\eps$ needed to cover $\cF$. The entropy number of $\cF$ is $\log N(\eps,\cF,Q)$. Then (under additional measurability assumptions, cf.~Remark~\ref{rem:measurability}) $\cF$ is GC if
\begin{align*}
\sup_{Q \in \cP(\cZ)} N(\eps,\cF,Q) < \infty, \qquad \forall \eps > 0.
\end{align*}
Other conditions for a class to be GC involve alternative notions of entropy, such as entropy with bracketing. Chapter~2 of van der Waart and Wellner \cite{WaaWel96} contains a detailed exposition of these matters. Examples~\ref{ex:VC}--\ref{ex:smooth} below follow \cite{WaaWel96}; Example~\ref{ex:BL} shows that the well-known theorem of Varadarajan on almost sure weak convergence of empirical measures can be stated in the form of a ULLN for an appropriate GC class.

\begin{example}[Vapnik--Chervonenkis classes]\label{ex:VC} {\em Given any collection $\cA \subset \cB_{\cZ}$ and any finite set $C \subset \cZ$, define
\begin{align*}
S(\cA, C) &\deq \left| \left\{ C \cap A : A \in \cA \right\} \right| \\
S_n(\cA) &\deq \max_{|C| \le n} S(\cA,C)
\end{align*}
and let $V(\cA) \deq \max \left\{ n \in \N : S_n(\cA) = 2^n \right\}$. After the fundamental work of Vapnik and Chervonenkis \cite{VC71} where these combinatorial parameters were first introduced, any class $\cA$ such that $V(\cA) < \infty$ is called a {\em Vapnik--Chervonenkis (VC) class}, and $V(\cA)$ is called its {\em Vapnik--Chervonenkis (VC) dimension}. Examples of VC classes include:
\begin{itemize}
	\item The class of all rectangles in $\R^d$ with VC dimension $2d$.
	\item The class of all linear halfspaces $H_{w,b} = \{ z \in \R^d : \langle w,z \rangle + b \ge 0 \}$ for $w \in \R^d$, $b \in \R$, with VC dimension $d+1$.
	\item The class of all closed balls $B_{x,r} = \{ z \in \R^d : \| z - x \| \le r \}$ for $x \in \R^d$, $r \in \R^+$, with VC dimension $d+1$.
\end{itemize}
Given a collection $\cA \subset \cB_{\cZ}$, let $\cF \equiv \cF_\cA$ consist of the indicator functions of the elements of $\cA$: $\cF_\cA = \{ 1_A : A \in \cA \}$. Then $\cF_\cA$ is GC, provided $\cA$ is a VC class.

Finite set-theoretic operations (unions, intersections, complements) on VC classes yield VC classes as well. In particular, consider the collection of all Voronoi cells induced by all $m$-point subsets of $\R^d$. Each member of this collection is an intersection of $O(m)$ halfspaces, and therefore we have a VC class. Likewise, injective images of VC classes are VC. 
}
\end{example}

\begin{example}[VC-subgraph classes]\label{ex:VC_sgr}
{\em Given a function $f \in M(\cZ)$, its {\em subgraph} is the subset of $\cZ \times \R$, given by $\{(z,t) : f(z) > t \}$. A class of functions $\cF \subset M(\cZ)$ is called a {\em VC-subgraph class} if the collection of all subgraphs of all $f \in \cF$ is a VC class in $\cZ \times \R$. We define $V(\cF)$, the VC dimension of $\cF$, as the VC dimension of the corresponding collection of subgraphs. For example, if $\cF$ is a linear span of $m$ functions $f_1,\ldots,f_m \in M(\cZ)$, then it is a VC-subgraph class with $V(\cF) \le m + 2$. In this paper, we are interested primarily in the case when $\cF \subset M^{b,1}(\cZ)$. Hence, if $f_1,\ldots,f_m \in M^{b,1}(\cZ)$, then their convex hull is a VC-subgraph class.}
\end{example}

\begin{example}[VC-hull classes]\label{ex:VC_hull} {\em A class of functions $\cF \subset M(\cZ)$ is a {\em VC-hull} class if there exists a VC-subgraph class $\cG \subset M(\cZ)$, such that every $f \in \cF$ is a pointwise limit of a sequence of functions $\{f_n\}$ contained in the {\em symmetric convex hull} of $\cG$,
	\begin{align*}
		\left\{ \sum^m_{i=1} c_i g_i : m \in \N;  \sum^m_{i=1}|c_i| \le 1; g_1, \ldots, g_m \in \cG \right\}
	\end{align*}
	For example, the set of all monotone functions $f : \R \to [0,1]$ is VC-hull (though not VC-subgraph).
}
\end{example}

\begin{example}[Smooth functions]\label{ex:smooth}
{\em Let $\cZ = [0,1]^d$. For any multi-index, i.e., a vector $\underline{k} = (k_1,\ldots,k_d) \in \{0,1,\ldots \}^d$, define the differential operator
	\begin{align*}
		D^{\underline{k}} \deq \frac{\partial^{|k|}}{\partial z_1^{k_1} \ldots \partial z_d^{k_d}},
	\end{align*}
where $|k| \deq k_1 + \ldots + k_d$. Given $\alpha > 0$, define for a function $f : [0,1]^d \to \R$
	\begin{align*}
		\| f \|_\alpha & \deq \max_{\underline{k}: |k| \le \lfloor \alpha \rfloor} \sup_{z} \left| D^{\underline{k}} f(z) \right| \nonumber\\
		& \qquad \qquad + \max_{\underline{k}: |k| = \lfloor \alpha \rfloor} \sup_{z \neq z'} \frac{\left|D^{\underline{k}}f(z) - D^{\underline{k}}f(z')\right|}{\| z - z' \|^{\alpha - \lfloor \alpha \rfloor}}
	\end{align*}
Let $C^\alpha$ be the set of all continuous functions $f : [0,1]^d \to \R$ with $\| f \|_\alpha \le 1$. Then $C^\alpha$ is a GC class.}
\end{example}

\begin{example}[Bounded Lipschitz functions]\label{ex:BL} {\em Let $(\cZ,d)$ be a complete separable metric space. Define the {\em Lipschitz seminorm} $\| \cdot \|_{\rm L}$ on $M(\cZ)$ by
	\begin{align*}
		\| f \|_{\rm L} \deq \sup_{z \neq z'} \frac{|f(z) - f(z')|}{d(z,z')}
	\end{align*}
	and the {\em bounded Lipschitz norm} $\| \cdot \|_{\rm BL}$ by
	\begin{align*}
		\| f \|_{\rm BL} \deq \| f \|_\infty + \| f \|_{\rm L}.
	\end{align*}
	Note that any function $f$ with $\| f \|_{\rm BL} < \infty$ is automatically in $C^b(\cZ)$, the Banach space of all bounded continuous functions on $\cZ$.

Let $\cF^1_{\rm BL} = \{ f \in C^b(\cZ) : \| f \|_{\rm BL} \le 1\}$. Then $\cF$ is a GC class. This is a consequence of the fact that the {\em bounded Lipschitz metric} (also known as the Fortet--Mourier metric)
\begin{align*}
	\beta(P,P') &\deq \sup_{ f \in \cF^1_{\rm BL}} |P(f) - P'(f)| \\
	& \equiv \| P - P' \|_{\cF^1_{\rm BL}} \quad P,P' \in \cP(\cZ)
	\end{align*}
	metrizes the topology of weak convergence in $\cP(\cZ)$. Recall that a sequence $\{P_n\}$ in $\cP(\cZ)$ converges {\em weakly} to $P \in \cP(\cZ)$ (the fact denoted by $P_n \rightsquigarrow P$) if
	\begin{align*}
		P_n(f) \xrightarrow{n \to \infty} P(f), \quad \forall f \in C^b(\cZ).
		\end{align*}
		Then $P_n \rightsquigarrow P$ if and only if $\beta(P_n,P) \xrightarrow{n \to \infty} 0$ \cite[Theorem~11.3.3]{Dud02}. Now, according to a theorem of Varadarajan \cite[Theorem~11.4.1]{Dud02}, given any i.i.d.\ random process $\{Z_i\}^\infty_{i=1}$ over $\cZ$ with common marginal distribution $P \in \cP(\cZ)$, the empirical distributions $\emp_{Z^n}$ converge weakly to $P$ almost surely:
		\begin{align}\label{eq:Varadarajan}
			\emp_{Z^n} \rightsquigarrow P \qquad \text{a.s.}
			\end{align}
			From the foregoing discussion, \eqref{eq:Varadarajan} is equivalent to
			\begin{align*}
				\beta(\emp_{Z^n},P) &= \sup_{f \in \cF^1_{\rm BL}} \left| \emp_{Z^n}(f) - P(f) \right| \\
				&\equiv \left\| \emp_{Z^n} - P \right\|_{\cF^1_{\rm BL}} \xrightarrow{n \to \infty} 0 \quad \text{a.s.}
				\end{align*}
In other words, $\cF^1_{\rm BL}$ is a GC class, and Varadarajan's theorem can be paraphrased to say that this function class obeys a ULLN.}
\end{example}

\section{Rethinking typicality for general alphabets}
\label{sec:general_typicality}

Now that all necessary definitions are made, we can introduce our revised notion of typicality for standard Borel spaces.

For finite alphabets, there are multiple equivalent definitions of a typical sequence. Here is one, based on the total variation distance \cite{CufPerCov09}, often referred to as {\em strong typicality} \cite[Section~10.6]{CovTho06}: 

\begin{definition}\label{def:strong_typicality} Given a finite set $\cZ$ and a probability distribution (mass function) $P$ on it, the {\em typical set} $\cT^{(n)}_\eps(P)$, for $\eps > 0$, is the set of all $n$-tuples $z^n \in \cZ^n$ whose empirical distributions $\emp_{z^n}$ are $\eps$-close to $P$ in total variation:
\begin{align*}
\cT^{(n)}_\eps(P) \deq \left\{ z^n \in \cZ^n: \dvar{\emp_{z^n}}{P} < \eps \right\}.
\end{align*}
\end{definition}

\noindent By the Law of Large Numbers, if $\{Z_i\}$ is a sequence of i.i.d.\ draws from $P$, then
\begin{align*}
	\Pr \big( Z^n \not\in \cT^{(n)}_\eps(P) \big) \xrightarrow{n \to \infty} 0.
\end{align*}
	If $\cZ$ is a Cartesian product $\cX \times \cY$, then one can define {\em jointly} and {\em conditionally} typical sets and sequences \cite{CovTho06}.

However, all of this breaks down for general (uncountably infinite) alphabets. The reason is that the total variation distance between any discrete measure and a nonatomic measure is equal to $2$. Indeed, if $(\cZ,\cB_\cZ)$ is a standard Borel space and $P \in \cP(\cZ)$ assigns zero mass to singletons, $P(\{ z\}) = 0, \forall z \in \cZ$, then we can take any $n$-tuple $z^n \in \cZ^n$ and let $A$ be the set of its {\em distinct} elements, so that $\emp_{z^n}(A) = 1$ and $P(A) = 0$. Using this and the definition \eqref{eq:dvar}, we deduce that $\dvar{\emp_{z^n}}{P} = 2$.

Of course, one could use typicality arguments by considering arbitrary finite quantizations of the underlying spaces, but, as long as we are dealing with nonatomic measures, this does not get rid of the above issue even in the limit of increasingly fine quantizations. While discretization is sufficient for many purposes \cite{Gra90a}, there is another issue that arises when dealing with Markov structures in multiterminal settings: quantization destroys the Markov property \cite[Section~VIII]{Csi98}.

To resolve this conundrum, we recall (cf.~Sec.~\ref{sec:prelims}) that
\begin{align*}
\dvar{P}{P'} = \sup_{\|f\|_\infty \le 1} |P(f) - P'(f)|,
\end{align*}
where the supremum is over {\em all} measurable functions $f : \cZ \to [-1,1]$. When the underlying measurable space supports nonatomic probability measures, this function class turns out to be too large to admit uniform convergence of empirical averages to statistical expectations. A natural solution, then, is to restrict the class of functions:

\begin{definition}\label{def:GC_typicality} Let $\cZ$ be a Borel space and let $\cF \subset M^{b,1}(\cZ)$ be a GC class of functions. Given a probability measure $P \in \cP(\cZ)$, the {\em typical set} $\cT^{(n)}_{\eps,\cF}(P)$, for $\eps > 0$, is the set of all $n$-tuples $z^n \in \cZ^n$ whose empirical distributions $\emp_{z^n}$ are $\eps$-close to $P$ in the $\| \cdot \|_\cF$ seminorm:
\begin{align*}
\cT^{(n)}_{\eps,\cF}(P) \deq \left\{ z^n \in \cZ^n: \| \emp_{z^n} - P \|_\cF < \eps \right\}.
\end{align*}
\end{definition}
\noindent One thing to note is that when $\cZ$ is finite, we can just take $\cF = M^{b,1}(\cZ)$ and immediately recover Definition~\ref{def:strong_typicality}. Moreover, if $\cZ$ is a complete separable metric space, then we can take $\cF = \cF^1_{\rm BL}$, in which case our notion of typicality becomes compatible with the bounded Lipschitz metric that metrizes the weak topology on the space of probability laws (cf.~Example~\ref{ex:BL}). 

\subsection{Basic properties of GC typical sets}

We now establish several basic properties of GC typical sets. First of all, any sufficiently long sequence emitted by a stationary memoryless source is typical with high probability:

\begin{proposition}\label{prop:GC_iid} Consider a Borel space $\cZ$ and a GC class $\cF \subset M^{b,1}(\cZ)$. If $\{Z_i\}^\infty_{i=1}$ is an i.i.d.\ random process over $\cZ$ with common law $P$, then for any $\eps > 0$
	\begin{align*}
			\lim_{n \to \infty}\Pr\big(Z^n \not\in \cT^{(n)}_{\eps,\cF}(P)\big) = 0
	\end{align*}
\end{proposition}
	
\begin{IEEEproof} Immediate from definitions.
\end{IEEEproof}

Another desirable property is for typicality to be preserved under coordinate projections. It is not hard to show that, for any two finite alphabets $\cX$ and $\cY$ and any two $n$-tuples $x^n \in \cX^n$ and $y^n \in \cY^n$ that are jointly typical w.r.t.\ some $P \in \cP(\cX \times \cY)$ in the sense of Definition~\ref{def:strong_typicality}, $x^n$ (resp., $y^n$) is typical w.r.t.\ the marginal distribution $P_X$ (resp., $P_Y$). The following lemma gives a sufficient condition for GC typicality to be preserved under projections:

\begin{proposition} Suppose $\cZ = \cX \times \cY$. Let $\pi_\cX : \cZ \to \cX$ be the coordinate projection mapping onto $\cX$, i.e., $\pi_\cX(x,y) = x$, and extend it to tuples via
\begin{align*}
\pi_\cX((x_1,y_1),\ldots,(x_n,y_n)) = (x_1,\ldots,x_n).
\end{align*}
Then for any $n \in \N$, any $\eps > 0$, any $P \in \cP(\cZ)$, and any GC class $\cF_\cX \subset M^{b,1}(\cX)$ such that $\cF_\cX \circ \pi_\cX \subseteq \cF$, we have the inclusion
\begin{align}\label{eq:GC_marginal}
	\pi_\cX \left( \cT^{(n)}_{\eps,\cF}(P)\right) \subseteq \cT^{(n)}_{\eps,\cF_\cX}(P_X).
\end{align}
\end{proposition}

\begin{remark}{\em As can be seen from the proof below, the class $\cF_\cX$ need not be GC in order for the inclusion \eqref{eq:GC_marginal} to hold. However, then one would not be able to transfer a convergence result like Proposition~\ref{prop:GC_iid} to the $\cX$-valued part of the sequence.}
\end{remark}

\begin{IEEEproof} Suppose $z^n = ((x_1,y_1),\ldots,(x_n,y_n)) \in \cT^{(n)}_{\eps,\cF}(P)$. Then
\begin{align*}
	& \left\| \emp_{x^n} - P_X \right\|_{\cF_\cX} \nonumber\\
	&= \sup_{f \in \cF_\cX} \left| \frac{1}{n}\sum^n_{i=1}f(x_i) - P_X(f) \right| \\
	&= \sup_{f \in \cF_\cX} \left| \frac{1}{n}\sum^n_{i=1}f \circ \pi_\cX(z_i) - P(f \circ \pi_\cX) \right| \\
	&\le \sup_{f \in \cF} \left| \frac{1}{n}\sum^n_{i=1}f(z_i) - P(f)\right| \\
	&= \left\| \emp_{z^n} - P \right\|_\cF \\
	& < \eps.
\end{align*}
Thus, $x^n \in \cT^{(n)}_{\eps,\cF_\cX}(P_X)$, which proves \eqref{eq:GC_marginal}.
\end{IEEEproof}

As an example, let $\cX = \R^k$, $\cY = \R^m$, let $\cF$ be the collection of indicator functions of all halfspaces in $\cZ = \R^{k+m}$, and let $\cF_\cX$ be the collection of indicator functions of all halfspaces in $\cX$ (cf.~Example~\ref{ex:VC} for definitions and notation). For any $w \in \R^k, b \in \R$ and $z = (x,y) \in \cZ$, we have
\begin{align*}
\langle w,x \rangle + b  &= \langle w,\pi_\cX(z) \rangle + b \\
&= \langle (w,0),(x,y)\rangle + (b,0).
\end{align*}
Hence, $1_{H_{(w,0),(b,0)}} = 1_{H_{(w,b)}} \circ \pi_\cX$ for any choice of $w \in \R^k,b \in \R$, so the condition of the lemma is satisfied.

Finally, we show that our definition of typicality can work in a multiterminal setting. Ideally, one would like to have something like the Markov lemma \cite{Ber78,Kra07}: If $X \to Y \to Z$ is a Markov chain, $(x^n,y^n)$ is typical, and $Z^n$ is obtained by passing $y^n$ through a memoryless channel, then $(x^n,y^n,Z^n)$ should be typical with high probability. However, in our setting such a statement does not make much sense without assuming additional structure for the function class $\cF$.\footnote{Incidentally, this is exactly what Mitran \cite{Mitran} accomplishes for his notion of typicality based on weak convergence.} Instead, we establish the following result, which is essentially an abstract alphabet version of the so-called Piggyback Coding Lemma of Wyner \cite[Lemma~4.3]{Wyn75}:

\begin{lemma}\label{lm:piggyback} Let $U \in \cU$, $V \in \cV$, and $W \in \cW$ be random variables taking values in their respective standard Borel spaces according to a joint distribution $P_{UVW}$, such that $U \to V \to W$ is a Markov chain and $I(V;W) < \infty$. Let $\{(U_i,V_i,W_i)\}^\infty_{i=1}$ be a sequence of i.i.d.\ draws from $P_{UVW}$. Let $\cF \subset M^{b,1}(\cU \times \cW)$ be a GC class of functions. For a given $\eps > 0$, there exist an $n = n(\eps)$ and a mapping $\Phi_n : \cV^n \to \cW^n$, such that
\begin{align}\label{eq:PB_1}
\frac{1}{n}\log\left| \left\{ \Phi_n(v^n) :  v^n \in \cV^n \right\} \right| \le I(V;W) + \eps
\end{align}
and
\begin{align}\label{eq:PB_2}
\Pr \left( (U^n, \Phi_n(V^n)) \not\in \cT^{(n)}_{\eps,\cF}(P_{UW}) \right) < \eps.
\end{align}
\end{lemma}

\begin{IEEEproof} For each $n$, define the function $\psi_n \in M^{b,1}(\cU^n \times \cW^n)$ by
	\begin{align*}
		\psi_n(u^n,w^n) \deq 1_{\left\{ (u^n, w^n) \not\in \cT^{(n)}_{\eps,\cF}(P_{UW}) \right\} }.
	\end{align*}
	Since $\cF$ is a GC class, we have by Proposition~\ref{prop:GC_iid}
	\begin{align*}
		\lim_{n \to \infty} \E \psi_n(U^n,W^n) = 0.
	\end{align*}
The desired statement now follows from Lemma~\ref{lm:PB} in Appendix~\ref{app:PB}.
\end{IEEEproof}

\section{Applications to empirical coordination}
\label{sec:example_apps}

We now show three sample applications of GC typicality to the problem of empirical coordination in a two-node network shown in Figure~\ref{fig:generic}. This problem, recently formulated and studied by Cuff et al.~\cite{CufPerCov09}, concerns joint generation of actions at the two nodes, such that the empirical distribution of the actions over time approximates, asymptotically, a desired joint distribution in total variation. Our goal is to extend this setting to general alphabets. As we have shown in Section~\ref{sec:general_typicality}, the total variation criterion is unsuitable for uncountable alphabets, so we consider a relaxation to an appropriate GC class.

As we will show, our notion of GC typicality and Lemma~\ref{lm:piggyback} can be used to develop particularly intuitive achievability arguments and to obtain single-letter characterizations of the best achievable rates. Moreover, convexity of the $\| \cdot \|_\cF$ seminorm is helpful for proving converse results. The downside, however, is that, in general, it is not possible to compute the best achievable rates explicitly even for ``simple'' sources due to the presence of the supremum over $\cF$.

\subsection{Two-node empirical coordination}
\label{ssec:coord}

\begin{figure}
	\centerline{\includegraphics[width=0.8\columnwidth]{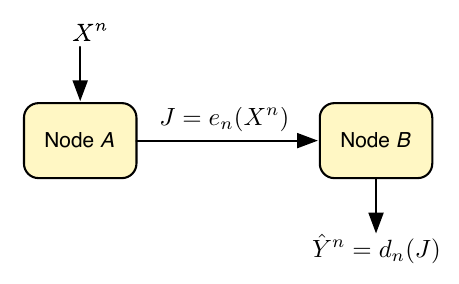}}
	\caption{\label{fig:coord}Two-node empirical coordination.}
\end{figure}

Consider the two-node network shown in Fig.~\ref{fig:coord}, where Node $A$ (resp.,~Node $B$) generates actions from a Borel space $\cX$ (resp.,~$\cY$). At Node $A$, the actions are drawn i.i.d.\ from a fixed law $P_X \in \cP(\cX)$. We also have a conditional probability measure $P_{Y|X}$ that describes the desired distribution of actions at Node $B$ given the actions at Node $A$. Following the terminology of \cite{CufPerCov09}, we will also refer to the choice of $P_{Y|X}$ as a {\em coordination}.  Node $A$ can communicate with Node $B$ over a rate-limited channel, and Node $B$ uses the data it receives to choose its actions. For each $n$, let $X^n \in \cX^n$ and $\wh{Y}^n \in \cY^n$ denote the action sequences at the two nodes. Given a class $\cF \subset M^{b,1}(\cX \times \cY)$ of measurable ``test functions" and a desired distortion level $\Delta \ge 0$, the goal is for Node $A$ to communicate with Node $B$ at a minimal rate to guarantee that, asymptotically,
\begin{align*}
\E \big\| \emp_{(X^n,\wh{Y}^n)} - P_X \otimes P_{Y|X} \big\|_\cF \lesssim \Delta,
\end{align*}
where $P_{XY} = P_X \otimes P_{Y|X}$ is the joint law induced by the source $P_X$ and the coordination $P_{Y|X}$. This is a generalization of the problem of {\em communication of probability distributions}, recently formulated and studied by Kramer and Savari \cite{KraSav07} in the finite-alphabet setting. Here, we allow general alphabets.

\begin{definition} An {\em $(n,M)$-code} is a pair $(e_n,d_n)$, where $e_n : \cX^n \to [M]$ is the {\em encoder} and $d_n : [M] \to \cY^n$ is the {\em decoder}, and $[M] \deq \{1,2,\ldots,M\}$. We will denote $\wh{Y}^n = d_n(e_n(X^n))$.
\end{definition}

\begin{definition} Given a source $P_X$, a coordination $P_{Y|X}$, and a distortion $\Delta$, let $\cE(\Delta,P_{Y|X})$ denote the set of all $Q \in \cP(\cX \times \cY)$, such that
	\begin{align*}
		Q_X = P_X \quad\text{and}\quad
\| Q - P_X \otimes P_{Y|X} \|_\cF \le \Delta.
\end{align*}
Define the {\em rate-distortion function for empirical coordination} as
\begin{align*}
R(\Delta,P_{Y|X}) \deq \inf_{Q \in \cE(\Delta,P_{Y|X})} I(Q).
\end{align*}
\end{definition}

\begin{theorem}\label{thm:coord}  Let $P_{Y|X}$ be a given coordination and $\Delta$ a given distortion level.
\begin{itemize}
\item [a)] {\bf Direct part:} If $\cF$ is a GC class and $R(\Delta,P_{Y|X}) < \infty$, then for any $\eps > 0$ there exist $n \equiv n(\eps)$ and an $(n,2^{nR})$ code $(e_n,d_n)$ with $R < R(\Delta,P_{Y|X}) + \eps$ satisfying
\begin{equation}\label{eq:coord_direct}
\E \big\| \emp_{(X^n,\wh{Y}^n)} - P_X \otimes P_{Y|X} \big\|_\cF \le \Delta + \eps.
\end{equation}
\item [b)] {\bf Converse part:} Suppose that there exists an $(n,2^{nR})$-code $\wh{Y}^n(X^n) = d_n(e_n(X^n))$, satisfying
\begin{equation}\label{eq:coord_converse}
\E \big\| \emp_{(X^n,\wh{Y}^n)} - P_X \otimes P_{Y|X} \big\|_\cF \le \Delta.
\end{equation}
Then $R \ge R(\Delta,P_{Y|X})$.
\end{itemize}
\end{theorem}

\begin{remark}{\em Note that the converse does not require $\cF$ to be GC. However, it must be sufficiently ``well-behaved'' for $\| \emp_{(X^n,\wh{Y}^n)} - P_X \otimes P_{Y|X}\|_\cF$ to be measurable for any choice of a (measurable) encoder-decoder pair.}\end{remark}

\begin{IEEEproof}[Proof (direct part)] To prove the direct part, fix $(\Delta,P_{Y|X})$ and pick any $Q \in \cE(\Delta,P_{Y|X})$ such that $I(Q) < R(\Delta,P_{Y|X}) + \eps/2$. Let $X \in \cX$ and $U \in \cY$ have joint law $Q$. Then $X \to X \to U$ is a Markov chain, and Lemma~\ref{lm:piggyback} guarantees the existence of an $n$ and a mapping $ \Phi_n : \cX^n \to \cY^n$, such that
\begin{align*}
\frac{1}{n} \log \left| \left\{ \Phi_n(\cX^n) \right\}\right| & \le I(Q) + \eps/2 < R(\Delta,P_{Y|X}) + \eps
\end{align*}
and
\begin{align*}
\E\left\| \emp_{(X^n,\Phi_n(X^n))} - Q \right\|_\cF \le \eps.
\end{align*}
Let $\wh{Y}^n = \Phi_n(X^n)$. Then the triangle inequality gives
\begin{align*}
& \E \big \| \emp_{(X^n,\wh{Y}^n)} - P_X \otimes P_{Y|X} \big\|_\cF \nonumber \\
& \qquad  \le \E \big\| \emp_{(X^n,\wh{Y}^n)} - Q \big\|_\cF + \big\| Q - P_X \otimes P_{Y|X} \big\|_\cF \\
& \qquad  \le \Delta + \eps,
\end{align*}
which establishes \eqref{eq:coord_direct}.\end{IEEEproof}

\begin{IEEEproof}[Proof (converse part)] For the converse, we will use the time mixing technique (cf.~\cite{CufPerCov09} and Appendix~\ref{app:time_mixing}). Let $\wh{Y}^n(X^n)$ be an $(n,2^{nR})$-code such that \eqref{eq:coord_converse} holds.  Let $T$ be a random variable uniformly distributed over the set $[n]$, independently of $X^n$, and let $\wh{Q}$ denote the joint distribution of $(X_T,\wh{Y}_T)$. Then
	\begin{align*}
		n R &\stackrel{{\rm(a)}}{\ge} H(\wh{Y}^n(X^n))  \\
		&= H(\wh{Y}^n(X^n)) - H(\wh{Y^n}(X^n)|X^n)  \nonumber\\
		&= I(X^n; \wh{Y}^n(X^n)) \nonumber \\
		&\stackrel{{\rm (b)}}{\ge} \sum^n_{t=1} I(X_t; \wh{Y}_t) \\
		&\stackrel{{\rm (c)}}{=} n I(X_T; \wh{Y}_T | T) \\
		&\stackrel{{\rm (d)}}{=} n I(X_T; \wh{Y}_T,T) \\
		&\ge n I(X_T; \wh{Y}_T) \nonumber \\
		&= n I(\wh{Q}),
	\end{align*}
where:
\begin{itemize}
	\item (a) holds because the log-cardinality of the range of $\wh{Y}^n(\cdot)$ is bounded by $nR$
	\item (b) is a standard information-theoretic fact: if $X^n$ is an i.i.d.\ tuple, then for any sequence $\wh{Y}_1,\ldots,\wh{Y}_n$ jointly distributed with $X^n$
\begin{align*}
	I(X^n; \wh{Y}^n) \ge \sum^n_{t=1} I(X_t; \wh{Y}_t)
\end{align*}
\item (c) follows from the construction of $T$
\item (d) holds because, by the chain rule for mutual information,
\begin{align*}
	I(X_T; \wh{Y}_T,T) = I(X_T;T) + I(X_T; \wh{Y}_T |T),
\end{align*}
where the first term on the r.h.s.\ is zero because $X^n$ is i.i.d.\ (see Fact 1 in Appendix~\ref{app:time_mixing}).
\end{itemize}
The remaining steps are consequences of other definitions and standard information-theoretic identities.

Since $X^n$ is i.i.d., $X_T$ is independent of $T$ and has the same distribution as $X_1$, namely $P_X$. Moreover, the expected empirical distribution $\E \emp_{(X^n,\wh{Y}^n)}$ is equal to $P_{(X_T,\wh{Y}_T)} \equiv \wh{Q}$ (Fact 2 in Appendix~\ref{app:time_mixing}). Thus, we can write
\begin{align*}
 \big \| \wh{Q} - P_X \otimes P_{Y|X} \big \|_\cF 
& = \big \| \E \emp_{(X^n,\wh{Y}^n)} - P_{X \otimes Y} \big \|_\cF \\
& \stackrel{{\rm (a)}}{\le} \E \big \| \emp_{(X^n,\wh{Y}^n)} - P_X \otimes P_{Y|X} \big \|_\cF \\
& \stackrel{{\rm (b)}}{\le} \Delta,
\end{align*}
where (a) follows from convexity, and (b) from \eqref{eq:coord_converse}. Hence, $\wh{Q} \in \cE(\Delta,P_{Y|X})$, so $R \ge I(\wh{Q}) \ge R(\Delta,P_{Y|X})$.
\end{IEEEproof}

\subsection{Two-node empirical coordination with side information}
\label{ssec:KS}

We now consider a generalization of the set-up from the preceding section, in which we also allow side information at the decoder. As before, we have a source distribution $P_X \in \cP(\cX)$ and a desired coordination $P_{Y|X}$. In addition, we have a side information channel $P_{Z|X}$ with input alphabet $\cX$ and output alphabet $\cZ$, which is also assumed to be standard Borel. Let $\{(X_i,Z_i)\}^\infty_{i=1}$ be an infinite sequence of independent draws from $P_{XZ} = P_X \otimes P_{Z|X}$. Consider the two-node network shown in Figure~\ref{fig:KS}. Node $A$ (resp.,~Node $B$) has perfect observations of $\{X_i\}$ (resp.,~$\{Z_i\}$). As before, Node $A$ can transmit information to Node $B$ over a rate-limited channel. The goal is for Node $A$ to communicate with Node $B$ at a minimal rate, so that Node $B$ can approximate the desired empirical process to within a given distortion level $\Delta$. More precisely, given a block length $n$ and denoting by $\wh{Y}^n$ the reconstruction of $Y^n$ at Node $B$, we wish to guarantee that
\begin{align*}
\E \big\| \emp_{(X^n,\wh{Y}^n)} - P_{XY} \big\|_\cF \lesssim \Delta.
\end{align*}
As we will see, the minimum achievable rate admits a single-letter characterization reminiscent of the Wyner--Ziv rate-distortion function for lossy source coding with decoder side information \cite{WynZiv76,Wyn78}. 

\begin{figure}
	\centerline{\includegraphics[width=0.8\columnwidth]{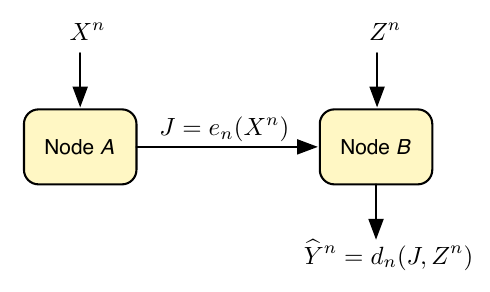}}
	\caption{\label{fig:KS}Two-node empirical coordination with side information.}
\end{figure}

\begin{definition} An {\em $(n,M)$-code} is a pair $(e_n,d_n)$, where  $e_n : \cX^n \to [M]$ is the {\em encoder} and $d_n : [M] \times \cZ^n \to \cY^n$ is the {\em decoder}. We will denote $\wh{Y}^n = d_n(e_n(X^n),Z^n)$.
\end{definition}

\begin{definition}Given a source $P_X$, a coordination $P_{Y|X}$, and a side information channel $P_{Z|X}$, let $\cE(\Delta,P_{Y|X},P_{Z|X})$ denote the set
	$$
	\{ Q \in \cP(\cX \times \cZ \times \cU) : \cU \text{ is standard Borel}\}
	$$
	such that:
\begin{enumerate}
\item $Q_{XZ} = P_{XZ}$
\item $Q_{U|XZ} = Q_{U|X}$ (i.e.,~$Z \to X \to U$ is a Markov chain)
\item There is a function $g : \cZ \times \cU  \to \cY$, such that
\begin{align*}
\| Q_{XW} - P_{XY} \|_\cF \le \Delta, 
\end{align*}
where $W = g(Z,U)$.
\end{enumerate}
With this, define the {\em rate-distortion function for empirical coordination with decoder side information} as
\begin{align*}
R(\Delta,P_{Y|X},P_{Z|X}) \deq \inf_{Q \in \cE(\Delta)} [I(Q_{XU}) - I(Q_{ZU})].
\end{align*}
\end{definition}

\begin{theorem}  Let $\cF$ be a class of functions $f : \cX \times \cY \to [0,1]$ and $\Delta$ a nonnegative distortion level.
\begin{itemize}
\item [a)] {\bf Direct part:} Suppose that $\cF$ is a GC class, and that for any $\delta > 0, \mu \in \cP(\cX \times \cY)$ one can find a finite set $\{ \wh{y}_j \}^N_{j=1} \subset \cY$ and a {\em quantizer} $q : \cY \to \{ \wh{y}_j \}$, such that
\begin{align}\label{eq:quantizer}
\| \mu_{Xq(Y)} - \mu \|_\cF \le \delta.
\end{align}
If  $R(\Delta,P_{Y|X},P_{Z|X}) < \infty$, then for any $\eps > 0$ there exist an $n \equiv n(\eps)$ and an $(n,2^{nR})$ code with $R < R(\Delta,P_{Y|X},P_{Z|X}) + \eps$ satisfying
\begin{equation}\label{eq:empirical_direct}
\E \big\| \emp_{(X^n,\wh{Y}^n)} - P_{XY} \big\|_\cF \le \Delta + \eps,
\end{equation}
where $\wh{Y}^n = d_n(e_n(X^n),Z^n)$.
\item [b)] {\bf Converse part:} Suppose that there exists an $(n,2^{nR})$-code $\wh{Y}^n = d_n(e_n(X^n),Z^n)$ satisfying
\begin{equation}\label{eq:empirical_converse}
\E \big\| \emp_{(X^n,\wh{Y}^n)} - P_{XY} \big\|_\cF \le \Delta.
\end{equation}
Then $R \ge R(\Delta,P_{Y|X},P_{Z|X})$.
\end{itemize}
\end{theorem}

\begin{remark} {\em The quantization assumption \eqref{eq:quantizer} is a ``smoothness" condition on $\cF$, and is akin to an assumption made by Wyner in \cite{Wyn78} in order to extend the achievability part of the finite-alphabet result of \cite{WynZiv76} to abstract alphabets.}
\end{remark}

\begin{IEEEproof}[Proof (direct part)] First we show that, owing to the quantization assumption \eqref{eq:quantizer}, we can assume w.l.o.g.\ that both $\cZ$ and the auxiliary alphabet $\cU$ are finite. This follows from the following lemma, whose proof is given in Appendix~\ref{app:WZ_approx}:
	
	\begin{lemma}\label{lm:WZ_approx}
Consider any law $Q \in \cE(\Delta,P_{Y|X},P_{Z|X})$. Then, for any $\delta > 0$, there exist finite measurable partitions $\{ A_i \}^{N_1}_{i=1}$ and $\{ B_j \}^{N_2}_{j=1}$ of $\cZ$ and $\cU$ and a function $g_1 : \cZ \times \cU \to \cY$ such that:
	\begin{enumerate}
	\item [a)] $\| Q_{XW_1} - P_{XY} \|_\cF \le \Delta + \delta$, where $W_1 = g_1(Z,U)$
	\item [b)] $g_1$ is constant on the rectangles $A_i \times B_j$, $1 \le i \le N_1, 1 \le j \le N_2$
	\item [c)] $I(Q_{X\td{U}}) - I(Q_{\td{Z}\td{U}}) \le I(Q_{XU}) - I(Q_{ZU}) + \delta$ where $\td{Z} = i$ for $Z \in A_i$ and $\td{U} = j$ for $U \in B_j$.
	\end{enumerate}
	\end{lemma}
	
Let us therefore assume that $\cU$ and $\cZ$ are both finite. We will use a Wyner--Ziv style two-step argument \cite{WynZiv76,Wyn78}: The first step consists of using a long block code that preserves typicality (following Lemma~\ref{lm:piggyback}), while the second step uses a Slepian--Wolf code \cite{SleWol73} to communicate the codewords with negligible probability of error. Pick any $Q \in \cE(\Delta,P_{Y|X},P_{Z|X})$ such that
\begin{align*}
	I(Q_{XU}) - I(Q_{ZU}) < R(\Delta,P_{Y|X},P_{Z|X}) + \eps/2.
\end{align*}
Define a function $\bar{g} : \cX \times \cZ \times \cU \to \cX \times \cY$ by $\bar{g}(x,z,u) \deq (x,g(z,u))$. Consider the function class $\cF \circ \bar{g} \subset M^{b,1}(\cX \times \cZ \times \cU)$. Since $\cF$ is a GC class, so is $\cF \circ \bar{g}$ --- to see this, fix any $\mu \in \cP(\cX \times \cZ \times \cU)$ and let $\{(X_i,Z_i,U_i)\}^\infty_{i=1}$ be a sequence of i.i.d.\ draws from $\mu$. Then for any $n$ we can write
\begin{align*}
&	\| \emp_{(X^n,Z^n,U^n)} - \mu \|_{\cF \circ \bar{g}} \nonumber\\
	 &= \sup_{f \in \cF} \left| \frac{1}{n}\sum^n_{i=1}f(X_i,g(Z_i,U_i)) - \E f(X,g(Z,U)) \right| \\
	&= \sup_{f \in \cF} \left| \frac{1}{n}\sum^n_{i=1} f(X_i,W_i) - \E f(X,W)\right| \\
	&\equiv \| \emp_{(X^n,W^n)} - \mu_{XW} \|_\cF,
\end{align*}
where $W = g(Z,U)$. Thus, the GC property of $\cF \circ \bar{g}$ follows from the GC property of $\cF$.\footnote{By contrast, in order for the GC property to be preserved under {\em left} compositions, i.e.,~for $\psi \circ \cF$ to be a GC class for some $\psi : [0,1] \to [0,1]$, additional requirements must be imposed on $\psi$ (such as monotonicity or Lipschitz continuity).} In view of this, we can apply Lemma~\ref{lm:piggyback} to the Markov chain $(X,Z) \to X \to U$ and to the GC class $\cF \circ \bar{g}$ to derive the existence of a large enough $n_1$ and a mapping $\Phi_{n_1} : \cX^{n_1} \to \cU^{n_1}$, such that
\begin{align*}
\frac{1}{n_1} \log |\{\Phi_{n_1}(\cX^{n_1})\}| \le I(Q_{XU}) + \eps/2
\end{align*}
and
\begin{align*}
& \E \big \| \emp_{(X^{n_1},Z^{n_1},\wh{U}^{n_1})} - Q_{XZU} \big\|_{\cF \circ \bar{g}} \nonumber\\
&\qquad = \E \big \| \emp_{(X^{n_1},\wh{W}^{n_1})} - Q_{XW} \big \|_\cF \le \eps/2,
\end{align*}
where
\begin{align*}
\wh{W}^{n_1} &= \big(g(Z_1,\wh{U}_1),\ldots,g(Z_{n_1},\wh{U}_{n_1})\big)\\
\wh{U}^{n_1} &= \Phi_{n_1}(X^{n_1}).
\end{align*}
We can use a blocking argument along the lines of Lemmas~3 and 5 of Wyner and Ziv \cite{WynZiv76} to show that a sufficiently long sequence $\wh{U}^{n_1}(1),\ldots,\wh{U}^{n_1}(n_2)$ of i.i.d.\ realizations of $\wh{U}^n$ can be losslessly encoded, using a Slepian--Wolf code, at a rate of
\begin{align*}
\frac{1}{n_1} H\big(\wh{U}^{n_1}\big|Z^{n_1}\big) &\le I(Q_{XU}) - I(Q_{ZU}) + \eps/2 \\
&< R(\Delta,P_{Y|X},P_{Z|X}) + \eps.
\end{align*}
Let $n=n_1n_2$, and let $\{\td{U}_i\}^{n}_{i=1}$ denote the resulting decoding. Then, if $n_2$ is large enough, we can guarantee that
\begin{align*}
	\E \big \| \emp_{(X^n,Z^n,\td{U}^n)} - \emp_{(X^n,Z^n,\wh{U}^n)} \big \|_{\cF \circ \bar{g}} \le \eps/2,
\end{align*}
and therefore, with $\wh{Y}^n = \big( g(Z_1,\td{U}_1),\ldots,g(Z_n,\td{U}_n)\big)$, that
\begin{align*}
& \E \big \| \emp_{(X^n,\wh{Y}^n)} - Q_{XW} \big \|_\cF \\
&\qquad = \E \big \| \emp_{(X^n,Z^n,\td{U}^n)} - Q_{XZU} \big\|_{\cF \circ \bar{g}} \\
&\qquad \le \E \big \| \emp_{(X^n,Z^n,\td{U}^n)} - \emp_{(X^n,Z^n,\wh{U}^n)} \big\|_{\cF \circ \bar{g}} \nonumber\\
&\qquad \qquad + \E \big\| \emp_{(X^n,Z^n,\wh{U}^n)} - Q_{XZU} \big\|_{\cF \circ \bar{g}} \\
&\qquad \le \eps.
\end{align*}
The triangle inequality then yields
\begin{align*}
& \E \big \| \emp_{(X^n,\wh{Y}^n)} - P_{XY} \big \|_\cF \nonumber \\
& \qquad \le
\E \big \| \emp_{(X^n,\wh{Y}^n)} - Q_{XW} \big \|_\cF + \big \| Q_{XW} - P_{XY} \big \|_\cF \\
& \qquad \le \Delta + \eps.
\end{align*}
Thus, we have constructed a $(n,2^{nR})$-code with rate $R < R(\Delta,P_{Y|X},P_{Z|X}) + \eps$.
\end{IEEEproof}

\begin{IEEEproof}[Proof (converse part)] To prove the converse, we again use time mixing. Let $(e_n,d_n)$ be an $(n,2^{nR})$ code, let $J = e_n(X^n)$ and $\wh{Y}^n = d_n(J,Z^n)$, and let $T$ be uniformly distributed on $[n]$ independently of $(X^n,Z^n)$. Define an auxiliary random variable
\begin{align*}
U = (J,X^{T-1},Z^{T-1},Z^n_{T+1},T)
\end{align*}
(cf.~\cite{WynZiv76,Wyn78,CufPerCov09}) and note that $Z_T \to X_T \to U$ is a Markov chain. Moreover,
\begin{align*}
	nR &\stackrel{{\rm (a)}}{\ge} H(J)  \\
	&\ge H(J|Z^n) \\
	&= I(X^n; J | Z^n) \\
	&= \sum^n_{t=1} I(X_t; J | Z^n, X^{t-1}) \\
	&\stackrel{{\rm (b)}}{=} \sum^n_{t=1} I(X_t; J,X^{t-1},Z^{t-1},Z^n_{t+1} | Z_t) \\
	&\stackrel{{\rm (c)}}{=} nI(X_T; J, X^{T-1},Z^{T-1},Z^n_{T+1} | Z_T,T) \\
	&\stackrel{{\rm (d)}}{=} nI(X_T; J, X^{T-1},Z^{T-1}, Z^n_{T+1}, T | Z_T) \\
	&= nI(X_T; U | Z_T),
\end{align*}
where:
\begin{itemize}
	\item (a) holds because the log-cardinality of the range of $e_n(\cdot)$ is bounded by $nR$
	\item (b) follows from the chain rule and the fact that $X_t \to Z_t \to (X^{t-1},Z^{t-1},Z^n_{t+1})$ is a Markov chain
	\item (c) follows from the construction of $T$
	\item (d) follows because, by the chain rule,
	\begin{align*}
		& I(X_T; J,X^{T-1},Z^{T-1},Z^n_{T+1}, T | Z_T)  \\
		&= I(X_T; T | Z_T) + I(X_T; J,X^{T-1},Z^{T-1},Z^{n}_{T+1} | Z_T,T)
	\end{align*}
	where the first term on the r.h.s.\ is zero because $(X_1,Z_1),\ldots,(X_n,Z_n)$ are i.i.d., so $(X_T,Z_T)$ is independent of $T$ (see Fact~1 in Appendix~\ref{app:time_mixing}).
\end{itemize}
The remaining steps are consequences of other definitions and standard information-theoretic identities.

Since $\{(X_i,Z_i)\}^n_{i=1}$ are i.i.d., $(X_T,Z_T)$ has the same joint law as $(X_1,Z_1)$, namely $P_{XZ}$. Moreover, $\wh{Y}_T$ is a deterministic function of $(Z_T,U)$, and $\E \emp_{(X^n,\wh{Y}^n)} = P_{(X_T,\wh{Y}_T)}$. Finally,
\begin{align*}
\big\| P_{(X_T,\wh{Y}_T)} - P_{XY} \big\|_\cF &= \big \| \E \emp_{(X^n,\wh{Y}^n)} - P_{XY} \big\|_\cF \\
& \stackrel{{\rm (a)}}{\le} \E \big\| \emp_{(X^n,\wh{Y}^n)} - P_{XY} \big\|_\cF  \\
&\stackrel{{\rm (b)}}{\le} \Delta,
\end{align*}
where (a) follows from convexity, and (b) follows from \eqref{eq:empirical_converse}. Hence, the joint law of $X_T$, $Z_T$, and $U$ belongs to $\cE(\Delta,P_{Y|X},P_{Z|X})$, which means that $R \ge I(X_T; U | Z_T) \ge R(\Delta,P_{Y|X},P_{Z|X})$.
\end{IEEEproof}

\subsection{Lossy coding with respect to a class of distortion measures}
\label{ssec:DW}

Finally, we consider the problem of lossy coding with respect to a class of distortion measures (fidelity criteria). For general (Polish) alphabets, it was solved by Dembo and Weissman \cite{DemWei03}, but the finite-alphabet variant appears already as Problem~14 in \cite{CsiKor81}. Let $\cX$ and $\cY$ denote the source and the reproduction alphabets, respectively. Suppose a class $\Gamma$ of distortion measures $\rho : \cX \times \cY \to [0,1]$ is given, together with a class of nonnegative reals indexed by $\rho \in \Gamma$, $\{\Delta_\rho\}_{\rho \in \Gamma}$. The goal is to find a block code of minimal rate whose expected distortion under each $\rho \in \Gamma$ is bounded by the corresponding $\Delta_\rho$. We use the same definition of an $(n,M)$-code as in Section~\ref{ssec:coord}.

Define a mapping $F(\cdot,\{\Delta_\rho\}) : \cP(\cX \times \cY) \to \R$ by
\begin{align*}
	F(Q,\{\Delta_\rho\}) \deq \sup_{\rho \in \Gamma} [Q(\rho) - \Delta_\rho],
\end{align*}
where
\begin{align*}
	Q(\rho) = \int \rho dQ = \int \rho(x,y) Q(dx,dy)
\end{align*}
is the expected distortion between $X$ and $Y$ when they have joint law $Q$.

\begin{definition}  Given a source $P_X \in \cP(\cX)$, let $\cE(\{\Delta_\rho\})$ denote the set of all $Q \in \cP(\cX \times \cY)$ such that
	\begin{align*}
		Q_X = P_X \quad\text{and}\quad F(Q,\{\Delta_\rho\}) \le 0.
		\end{align*}
		Define the {\em rate-distortion function}
\begin{align*}
R(\{\Delta_\rho\}) \deq \inf_{Q \in \cE(\{\Delta_\rho\})} I(Q).
\end{align*}
\end{definition}

Theorem~1 of \cite{DemWei03} shows that any rate $R \ge R(\{\Delta_\rho\})$ is achievable, provided the mapping $Q \mapsto F(Q,\{\Delta_\rho\})$ is upper semicontinuous (u.s.c.) under the weak topology on $\cP(\cX \times \cY)$. Moreover, no rate $R < R(\{\Delta_\rho\})$ is achievable. We now show that the u.s.c.\ requirement can be replaced by a GC condition:

\begin{theorem}  Let $\Gamma$ be a class of distortion measures and $\{\Delta_\rho\}_{\rho \in \Gamma}$ a class of nonnegative distortion levels.
\begin{itemize}
\item [a)] {\bf Direct part:} If $\Gamma$ is a GC class and $R(\{\Delta_\rho\}) < \infty$, then for any $\eps > 0$, there exist an $n \equiv n(\eps)$ and an $(n,2^{nR})$ code with $R < R(\{\Delta_\rho\}) + \eps$ satisfying
\begin{equation}\label{eq:manydist_direct}
\E \sup_{\rho \in \Gamma} \left[ \rho(X^n, \wh{Y}^n) - \Delta_\rho \right] \le  \eps,
\end{equation}
where $\rho(X^n,\wh{Y}^n) \deq \emp_{(X^n,\wh{Y}^n)}(\rho)$.
\item [b)] {\bf Converse part:} Suppose that there exists an $(n,2^{nR})$-code $\wh{Y}^n = d_n(e_n(X^n))$ satisfying
\begin{equation}\label{eq:manydist_converse}
\E \rho(X^n,\wh{Y}^n) \le \Delta_\rho, \qquad \forall \rho \in \Gamma.
\end{equation}
Then $R \ge R(\{\Delta_\rho\})$.
\end{itemize}
\end{theorem}

\begin{IEEEproof} To prove the direct part, pick any $Q \in \cE(\{\Delta_\rho\})$ such that $I(Q) < R(\{\Delta_\rho\}) + \eps/2$. Let $X \in \cX$ and $U \in \cY$ have joint law $Q$. The same argument as in the proof of Theorem~\ref{thm:coord} can be used to show the existence of a large enough $n$ and a mapping $\Phi_n : \cX^n \to \cY^n$, such that
\begin{align*}
\frac{1}{n} \log|\{\Phi_n(\cX^n)\}| &\le I(Q) + \eps/2 < R(\{\Delta_\rho\}) + \eps
\end{align*}
and
\begin{align*}
\E \big\| \emp_{(X^n,\wh{Y}^n)} - Q \big \|_\Gamma \le \eps,
\end{align*}
where $\wh{Y}^n = \Phi_n(X^n)$. Now, for any $\rho \in \Gamma$ we have
\begin{align*}
&\rho(X^n,\wh{Y}^n) - \Delta_\rho \le \| \emp_{(X^n,\wh{Y}^n)} - Q \|_\Gamma + F(Q,\{\Delta_\rho\}).
\end{align*}
Consequently, taking the supremum of both sides over $\Gamma$ and then the expectation w.r.t.\ $P_{X^n}$, we get (\ref{eq:manydist_direct}).

The proof of the converse is exactly the same as in \cite{DemWei03}.
\end{IEEEproof}

\section{Conclusion}
\label{sec:conclusion}

We have proposed a new definition of typical sequences over a wide class of abstract alphabets (standard Borel spaces), which retains many useful properties of strong (total-variation) typicality for finite alphabets. In particular, it is preserved in a Markov structure, which has allowed us to develop transparent achievability proofs in several settings pertaining to empirical coordination of actions in a two-node network using finite communication resources. Here are some directions for future research:
\begin{itemize}
	\item {\em Behavior in the finite block length regime} --- GC classes with sufficiently ``regular'' metric or combinatorial structure admit sharp concentration-of-measure inequalities of the form
	\begin{align*}
		\Pr\left( \| \emp_{Z^n} - P \|_\cF \ge \eps\right) \le S(n;\cF) e^{-Cn\eps^2},
	\end{align*}
	where $C > 0$ is some constant and $S(n;\cF)$ is a function of ``moderate'' growth in $n$, which typically depends on the geometric characteristics of $\cF$ \cite{Pol84,WaaWel96,vanDeGeer00}. For example, if $\cF$ is a VC class, then $S(n;\cF) = O(n^{V(\cF)})$; in the latter case, we also have
	\begin{align*}
		\E \big\| \emp_{Z^n} - P \big\|_\cF \le C\sqrt{\frac{V(\cF)}{n}},
	\end{align*}
	where $C > 0$ is a universal constant. These inequalities can be used to investigate the behavior of our coding schemes in the finite block length regime (e.g.,~the rate of convergence of the achievable $\| \cdot \|_\cF$-distortion to the optimum).
	\item {\em Extension to stationary ergodic sources} --- Recently, Adams and Nobel \cite{AdaNob10} have shown that the ULLN holds for countable (or separable) classes of VC sets and functions even when the underlying process is stationary and ergodic (rather than i.i.d.), although without any specific guarantees on the rate of convergence. Their work opens the possibility of extending our GC typicality approach to stationary ergodic sources via sliding block codes \cite{Dun80,Kie80,Kie81}.
	\item {\em Connections to simulation of information sources} --- The operational criteria used in our treatment of empirical coordination suggest new ways of thinking about simulation of random processes and related problems in rate-distortion coding \cite{HanVer93,SteVer96,MGL10,CufPerCov09}. Many problems related to sensing, learning, and control under communication constraints can be reduced (or related) to simulation of random processes, and our formalism may be of use for characterizing the fundamental information-theoretic limits in these settings.
\end{itemize}

\begin{appendices}
	\renewcommand{\theequation}{\Alph{section}.\arabic{equation}}
	\setcounter{lemma}{0}
	\setcounter{equation}{0}

	\renewcommand{\thelemma}{\Alph{section}.\arabic{lemma}}

\section{Piggyback Coding Lemma for Borel spaces}
\label{app:PB}

In this appendix we prove the following lemma, which is an extension of the Piggyback Coding lemma of Wyner \cite[Lemma~4.3]{Wyn75} to general alphabets:

\begin{lemma}\label{lm:PB} Let $\cU,\cV,\cW$ be standard Borel spaces, and let $(U,V,W) \in \cU \times \cV \times \cW$ be a triple of random variables with joint law $P_{UVW}$, such that $U \to V \to W$ is a Markov chain and the mutual information $I(V; W)$ is finite. Let $\{(U_i,V_i,W_i)\}^\infty_{i=1}$ be a sequence of i.i.d.\ draws from $P_{UVW}$. Let $\{ \psi_n \}^\infty_{n=1}$ be a sequence of measurable functions $\psi_n : \cU^n \times \cW^n \to [0,1]$, such that
\begin{align*}
\lim_{n \to \infty} \E \psi_n(U^n,W^n) = 0.
\end{align*}
For a given $\eps > 0$, there exists $n_0 = n_0(\eps)$, such that for every $n \ge n_0$ we can find a mapping $F_n: \cV^n \to \cW^n$ that satisfies
\begin{align*}
\frac{1}{n} \log \left| \Big\{ F_n(v^n) : v^n \in \cV^n \Big\} \right| \le I(V;W) + \eps
\end{align*}
and
\begin{align*}
\E \psi_n(U^n,F_n(V^n)) \le \eps.
\end{align*}
\end{lemma}

\begin{proof} The proof is very similar to Wyner's proof for finite alphabets \cite{Wyn75}. Fix any $n$ and define a function $\phi_n : \cV^n \times \cW^n \to [0,1]$ by
\begin{align*}
\phi_n(v^n,w^n) &\deq \E\Big[ \psi_n(U^n,W^n) \Big| V^v = v^n, W^n = w^n \Big] \\
&= \int_{\cU^n} \psi_n(u^n,w^n) P_{U^n|V^n,W^n}(du^n|v^n,w^n).
\end{align*}
Owing to the Markov chain condition, we can write
\begin{align}\label{eq:PB_Markov}
\phi_n(v^n,w^n) = \int_{\cU^n} \psi_n(u^n,w^n) P_{U^n|V^n}(du^n|v^n).
\end{align}
Letting $\delta_n \deq \E \psi_n(U^n,W^n)$, we define the set
\begin{align*}
\cS_n \deq \Big\{ (v^n,w^n) \in \cV^n \times \cW^n : \phi_n(v^n,w^n) \le \sqrt{\delta_n} \Big\}.
\end{align*}
Then by the Markov inequality we have
\begin{align*}
\Pr \left( (V^n, W^n) \not\in\cS_n \right) \le \frac{\E \phi_n(V^n,W^n)}{\sqrt{\delta_n}} = \sqrt{\delta_n}.
\end{align*}
Consider an arbitrary measurable mapping $G : \cV^n \to \{ w^n(1),\ldots,w^n(M)\} \subset \cW^n$ for some $M < \infty$. Then, defining the set
\begin{align*}
\tilde{\cS}_n \deq \{ v^n \in \cV^n : (v^n,G(v^n)) \in \cS_n\},
\end{align*}
we can write
\begin{align*}
&\E \psi_n(U^n, G(V^n)) \\
&= \E \big[ \E [\psi_n(U^n,G(V^n)) | V^n] \big] \\
&\stackrel{{\rm (a)}}{=} \E \phi_n (V^n, G(V^n))\\
&\stackrel{{\rm (b)}}{\le}  \Pr ( \tilde{\cS}^c_n )  + \int_{\tilde{\cS}_n} \phi_n(v^n,G(v^n)) P_{V^n}(dv^n), 
\end{align*}
where (a) is due to \eqref{eq:PB_Markov}, while (b) uses the fact that $0 \le \phi_n(\cdot,\cdot) \le 1$. Moreover,
\begin{align*}
\int_{\tilde{\cS}_n} \phi_n(v^n,G(v^n)) P_{V^n}(dv^n) 
&\le \sqrt{\delta_n}.
\end{align*}
Hence,
\begin{align*}
\E \psi_n (U^n, G(V^n)) \le \Pr ( \tilde{\cS}^c_n ) + \sqrt{\delta_n}.
\end{align*}
Now we can use Lemma~9.3.1 in \cite{Gal68} to show that, given $\cS_n$, $M$, and an arbitrary $R > 0$, there exist a set $\{ w^n(1),\ldots,w^n(M)\} \subset \cW^n$ and a mapping $G_n : \cV^n \to \{ w^n(1),\ldots,w^n(M) \}$, such that
\begin{align*}
& \Pr \left( (V^n,G_n(V^n)) \not\in \cS_n \right) \le \Pr (  \cS^c_n ) \nonumber\\
& \quad \quad + \Pr \left(i(V^n,W^n) > nR \right)  + \exp \left( - M2^{-Rn} \right),
\end{align*}
where
\begin{align*}
i(v^n,w^n) \deq \log \frac{dP_{V^n,W^n}}{d(P_{V^n} \otimes P_{W^n})}(v^n,w^n)
\end{align*}
is the information density \cite{Gra90a}. Letting $M = 2^{n(I(V;W) + \eps)}$ and $R = I(V;W) + \eps/2$ and using the corresponding mapping $G_n$, we get
\begin{align*}
& \E \psi_n(U^n,G_n(V^n))  \le 2\sqrt{\delta_n} \nonumber\\
& \quad \quad  + \exp(-2^{n\eps/2}) + \Pr \left( i(V^n,W^n) > nR \right).
\end{align*}
Since $\E \psi_n (U^n,W^n) = \delta_n \to 0$ as $n \to \infty$, the first term goes to zero as $n \to \infty$. The second term likewise goes to 0 since $\eps > 0$. The third term goes to zero owing to the mean ergodic theorem for information densities \cite[Theorem~8.5.1]{Gra90a}. Choosing $n_0$ large enough so that the right-hand side of the above inequality is less than $\eps$ finishes the proof.
\end{proof}

\section{Time mixing}
\label{app:time_mixing}

\setcounter{equation}{0}

Our discussion of the time mixing technique essentially follows \cite[p.~4200]{CufPerCov09}, except that care must be taken due to the fact that we are working with general alphabets here.

Fix a space $\cU$. Let $U^n = (U_1,\ldots,U_n)$ be a random $n$-tuple taking values in $\cU^n$ according to some law $P_{U^n}$. Let $T$ be a random variable uniformly distributed over the set $[n]$ independently of $U^n$. Consider the random variable $U_T \in \cU$, i.e.,~the value of the $T$th coordinate of $U^n$. We will use two facts pertaining to this construction.

First, we note that $U_T$ and $T$ need not be independent, even though $U^n$ and $T$ are. One exception is when $U^n$ is an i.i.d.\ tuple:

\begin{fact} If $U^n$ is an i.i.d.\ tuple with common marginal $P_U$, then $U_T$ is independent of $T$
and has the same law as $U_1$, i.e., $P_U$.\end{fact}

\begin{IEEEproof} For any $i \in [n]$ and any $A \in \cB_\cU$,
	\begin{align*}
		P_{U_T,T}(A \times \{i\}) &= \Pr(T=i) P_{U_T|T}(A|i) \\
		&= \Pr(T=i) P_{U_i}(A) \\
		&= \Pr(T=i) P_{U}(A) \\
		&= P_T(\{i\}) P_{U}(A).
	\end{align*}
	Hence, $P_{U_T|T}(A|i) = P_U(A)$, regardless of $i$.
\end{IEEEproof}

Second, let us consider the empirical distribution $\emp_{U^n}$. Since $\cU$ is a Borel space, $\cP(\cU)$ is a (complete separable) metric space under any metric that metrizes the weak convergence of probability laws, so we can equip it with its Borel $\sigma$-algebra. Then $\emp_{U^n}$ is a $\cP(\cU)$-valued random variable, whose expectation $\E \emp_{U^n}$ is given by
\begin{align*}
	[\E \emp_{U^n}](A) \deq \frac{1}{n} \sum^n_{i=1} P_{U_i}(A), \qquad \forall A \in \cB_\cU.
\end{align*}
It is not hard to check that $\E \emp_{U^n}$ satisfies the Kolmogorov axioms and is itself an element of $\cP(\cU)$. In particular:

\begin{fact} Consider the empirical distribution $\emp_{U^n}$. Then
	\begin{align}\label{eq:mixing_expect}
		\E \emp_{U^n} = P_{U_T},
	\end{align}
	where $P_{U_T} \in \cP(\cU)$ is the law of $U_T$.
\end{fact}

\begin{IEEEproof} For any $A \in \cB_\cU$,
	\begin{align*}
		[\E\emp_{U^n}](A) &= \frac{1}{n}\sum^n_{i=1} P_{U_i}(A) \\
		&= \E \left[\sum^n_{i=1} \Pr(T=i) 1_{\{U_i \in A\}}\right] \\
		&= \E \left[ \E \big[1_{\{U_T \in A\}}\big| U^n \big]\right] \\
		&= \E \left[ 1_{\{U_T \in A\}}\right] \\
		&= P_{U_T}(A).
	\end{align*}
	Since $A$ is arbitrary, \eqref{eq:mixing_expect} indeed holds.
\end{IEEEproof}

\section{Proof of Lemma~\ref{lm:WZ_approx}}
\label{app:WZ_approx}
\setcounter{equation}{0}

The proof is very similar to the proof of Lemma~5.3 of Wyner \cite{Wyn78}. In particular, only part (a) requires modification. Parts (b) and (c) follow immediately, just as in \cite{Wyn78}.

Since $Q \in \cE(\Delta)$, there exists a function $g : \cZ \times \cU \to \cY$, such that, with $W = g(Z,U)$,
\begin{align}\label{eq:aux_map}
	\big\| Q_{XW} - P_{XY} \big\|_\cF \le \Delta.
\end{align}
Secondly, owing to the smoothness assumption \eqref{eq:quantizer}, for any $\delta_1 > 0$ one can find a quantizer $q : \cY \to \{\wh{y}_j\}^N_{j=1} \subset \cY$, $N < \infty$, such that
\begin{align}\label{eq:quant}
	\big\| Q_{Xq(W)} - Q_{XW} \big\|_\cF \le \delta_1.
\end{align}
Let $g_0 \deq q \circ g$, and define the sets
\begin{align*}
	C_j \deq \left\{ (z,u) \in \cY \times \cU : g_0(z,u) = \wh{y}_j \right\}, \quad 1 \le j \le N.
\end{align*}
Lemma~5.4 in \cite{Wyn78} can be used to show that, for an arbitrary $\delta_2 > 0$, there exists a collection of disjoint sets $\{S_j\}^N_{j=1} \subset \cB_\cZ \otimes \cB_\cU$, where each $S_j$ is a finite union of rectangles, and
\begin{align}\label{eq:set_approx}
	Q_{ZU}(S_j \triangle C_j) \le \delta_2, \quad 1 \le j \le N.
\end{align}
Now define $g_1 : \cZ \times \cU \to \cY$ by
\begin{align*}
	g_1(y,u) \deq
	\begin{cases}
		\wh{y}_j, & \text{if } (z,u) \in S_j \\
		\wh{y}_1, & \text{if } (z,u) \not\in \bigcup^N_{j=1}S_j.
	\end{cases}
\end{align*}
Define also the set $E \deq \bigcup^N_{j=1} (C_j \cap S_j)$ and note that $g_1 = g_0$ on $E$. Then
\begin{align}
	&\E [f(X,g_1(Z,U))] \nonumber\\
	&= \E[1_E f(X,g_0(Z,U))] + \E[1_{E^c} f(X,g_1(Z,U))] \nonumber\\
	&\le \E[ f(X,g_0(Z,U))] + Q_{ZU}(E^c) \nonumber \\
	&= \E[f(X,q(W))] + Q_{ZU}(E^c) \nonumber\\
	&\le \E[f(X,W)] + \delta_1 + Q_{ZU}(E^c). \label{eq:WZ_approx_1}
\end{align}
Similarly,
\begin{align}
	&\E[f(X,W)] \nonumber\\
	& \le \E[f(X,q(W))] + \delta_1 \nonumber \\
	&= \E[1_E f(X,q(W))] + \E[1_{E^c} f(X,q(W))] + \delta_1 \nonumber \\
	&= \E[1_E f(X,g_1(Z,U))] + \E[1_{E^c} f(X,q(W))] + \delta_1 \nonumber \\
	&\le \E[f(X,g_1(Z,U))] + Q_{ZU}(E^c) + \delta_1. \label{eq:WZ_approx_2}
\end{align}
In both cases we have used the fact that $f$ is bounded between $0$ and $1$, as well as  \eqref{eq:quant}. Moreover, using the fact that $\{C_j\}$ is a disjoint partition of $\cZ \times \cU$, as well as \eqref{eq:set_approx}, we can write
\begin{align*}
	Q_{ZU}(E^c) \le
	\sum^N_{j=1} Q_{ZU}(S_j \triangle C_j) \le N\delta_2.
\end{align*}
Combining \eqref{eq:aux_map}, \eqref{eq:WZ_approx_1} and \eqref{eq:WZ_approx_2}, we get
\begin{align*}
	\big\| Q_{XW_1} - Q_{XW} \big\|_\cF \le \delta_1 + N \delta_2,
\end{align*}
where $W_1 = g_1(Z,U)$. Now, given $\delta > 0$, first choose $\delta_1 = \delta/2$. This fixes $N = N(\delta)$. Then choose $\delta_2$ so that $N\delta_2 \le \delta/2$. This proves part (a); parts (b) and (c) follow exactly as in \cite{Wyn78}.

\end{appendices}

\section*{Acknowledgment}

The author would like to thank Todd Coleman and Serdar Y\"uksel for their careful reading of the manuscript and for making a number of useful suggestions that have improved the presentation. Insightful comments by the Associate Editor Yossef Steinberg and two anonymous referees are also gratefully acknowledged. In particular, the author is indebted to one of the referees for pointing out a flaw in the original version of the problem formulation in Section~\ref{ssec:KS}.

\bibliographystyle{IEEEtran}
\bibliography{GC_typical_final.bbl}

\begin{IEEEbiographynophoto}{Maxim Raginsky} (S'99--M'00) received the B.S. and M.S. degrees in 2000 and the Ph.D. degree in 2002 from Northwestern University, Evanston, IL, all in electrical engineering. He has held research positions with Northwestern, the University of Illinois at Urbana-Champaign (where he was a Beckman Foundation Fellow from 2004 to 2007), and Duke University. In 2012, he has returned to UIUC, where he is currently an Assistant Professor with the Department of Electrical and Computer Engineering and the Coordinated Science Laboratory. His research interests lie at the intersection of information theory, machine learning, and control.\end{IEEEbiographynophoto}

\end{document}